\documentclass{article}
\usepackage[utf8]{inputenc}
\usepackage{igor-commands}
\usepackage[capitalise]{cleveref}
\usepackage{placeins}
\usepackage{authblk}

\usepackage[ shadow, colorinlistoftodos,textsize=tiny, obeyFinal]{todonotes}
\setuptodonotes{fancyline, backgroundcolor=gray!10,bordercolor=gray}

\newtheorem{thm}{Theorem}[section]
\newtheorem{lem}[thm]{Lemma}
\newtheorem{defn}[thm]{Definition}
\newtheorem{prop}[thm]{Proposition}
\newtheorem{cor}[thm]{Corollary}

\theoremstyle{definition}
\newtheorem{eg}[thm]{Example}
\newtheorem{rmk}[thm]{Remark}

\title{Lights Out Puzzle in $p$ Colors: Evolution of Quiet Patterns}
\author{Wisdom Boinde, Igor Minevich, and Dipesh Poudel}
\date{\today}
\affil{Wentworth Institute of Technology}

\DeclareMathOperator{\Era}{Era}

\begin{document}

\maketitle

\begin{abstract}
    The Lights Out Puzzle represents a cellular automaton based on a grid of squares where clicking a square changes its state and the states of surrounding squares. A ``quiet pattern'' is a way to click such that in the end, no change is effected. We introduce a way to ``evolve'' quiet patterns in smaller grids into ones in $p$ times larger grids when the number of possible states of a square is a prime $p$. Using elliptic curves, we also find that an inverse ``de-evolution'' exists for most $p$. We also describe the only ways to click a grid of squares such that only 5 (the minimum) number of squares have a nonzero state.
\end{abstract}

\section{Introduction}
The Lights Out Puzzle began as an electronic game made in 1995 by Tiger Electronics, where a more-or-less random stored pattern of lights in a $5 \times 5$ grid would be turned on. Pressing any button would switch the state of the light on that button (between on and off), but also the lights horizontally and vertically adjacent to it, and the player needed to turn all the lights off. Since then, several variants of the puzzle have been explored, not only to sizes $M \times N$ \citep{Jaap} but to various graphs and numbers of colors (rather than just on and off) \citep{LOFiniteGraphs, Kreh}, where the rules dictate that each vertex has a ``state'' (i.e. color) in $\Z/k\Z$, where $k$ is the number of colors, and when a vertex is clicked, its own state and the state of the adjacent vertices are changed by $+1 \pmod k$. The reader can try out the puzzle and many such variants online at \url{https://tiny.one/lights}, \url{https://www.jaapsch.net/puzzles/lights.htm#java}, and many other places. Recently, we have also started exploring the generalization where the state of a vertex is represented by an element of a group $G$, and clicking a vertex with a specified element $g \in G$ multiplies the states of it and adjacent vertices by $g$ on the right \citep{FirstPaper, SecondPaper}. In these papers, we have shown many graphs, including grid graphs, cylinders, and tori, as well as any possible higher-dimensional generalization of these except for odd-dimensional cubes essentially reduce to working with an abelian group $G$, which, if finitely generated, can itself be broken down into studying the puzzle over $\Z$ and $\Z/k\Z$ for prime-power $k$.  

In essence, in all these settings, a ``click'' is a linear operation that adds 1's to the states of certain vertices. More precisely, we can collect the states of all the squares of the puzzle (or all the vertices if we want to be most general) in a \emph{state vector}. Then  a ``click'' is simply a linear operation that adds a vector of 1's and 0's to that state vector. If we put the vectors of 1's and 0's that represents clicking the $j$-th vertex in the $j$-th column of a matrix $A$, then the effect of a \emph{click vector} $\vec v$ (representing the vertices we want to click and how many times we want to click them) is simply to add $A\vec v$ to the current state vector. The reader will note that this matrix $A = A_{M, N}$, which we call the \emph{activation matrix} for the $M \times N$ puzzle, is just the adjacency matrix of the graph plus the identity matrix (since vertices affect themselves). In particular, the Lights Out question asks if we can find a vector $\vec v$ such that $A\vec v$ such that $\vec 1 + A\vec v = \vec 0$ (where $\vec 1$ and $\vec 0$ are the vectors of all 1's and all 0's, respectively).

This is obviously equivalent to solving the matrix equation $A\vec v = \vec 1$, which always has a solution when there are 2 colors \citep{Sutner1989}, but not always when there are more colors, e.g. there is no solution with $k = 3$ to the $2 \times 2$ puzzle. The number of solutions is the same as the number of vectors in the null space of $A$, i.e. the number of vectors $\vec v$ for which $A\vec v = \vec 0$. We call the latter vectors \emph{quiet patterns}, as Jaap did in \citep{Jaap}, because adding the clicks in a quiet pattern to any preset clicks does not change the resulting states of the board; we can always quietly sneak these in. 

\eg The matrix $A$ below is the activation matrix $A_{2, 3}$ for the $2 \times 3$ grid, if we order the grid by going through the first row (left to right) and then the second row.
\[A = \begin{pmatrix}
    1 & 1 & 0 & 1 & 0 & 0\\
    1 & 1 & 1 & 0 & 1 & 0\\
    0 & 1 & 1 & 0 & 0 & 1\\
    1 & 0 & 0 & 1 & 1 & 0\\
    0 & 1 & 0 & 1 & 1 & 1\\
    0 & 0 & 1 & 0 & 1 & 1\\
\end{pmatrix}, \qquad \vec q = \vecsix111010, \qquad \vec r = \vecsix110110, \qquad \vec s = \vecsix100001\]
Now the vector $\vec q$ above forms a quiet pattern (the leftmost in Figure \ref{fig:qp2x3}) since $A\vec q = \vec 0$; $\vec s$ forms a solution vector for the 2-color puzzle since $A\vec s = \vec 1 \pmod 2$; and $\vec r$ gives a solution vector to $A\vec v = \vec 1$ for any $k$-color puzzle. In Figure \ref{fig:qp2x3}, we have displayed the nontrivial quiet patterns for the 2-color $2 \times 3$ Lights Out Puzzle, not as $6 \times 1$ vectors but as $2 \times 3$ grids where the numbers denote how many times each square is clicked. 

% The reader will note that, of course, the null space is closed under addition and, together with the trivial quiet pattern $\vec 0$, the total number of quiet patterns (including $\vec 0$) is a power of $k$ when $k$ is prime, namely $k^{n(M, N)}$, where $n(M, N)$ is the nullity of the activation matrix $A = A_{M, N}$ corresponding to the $M \times N$ puzzle.

\begin{figure}[hbtp]
    \centering
\begin{tabular}{ccc}

\begin{tabular}{|c|c|c|}
\hline
    1 & 1 & 1 \\
\hline
    0 & 1 & 0\\
\hline
\end{tabular} & 

\begin{tabular}{|c|c|c|}
\hline
    0 & 1 & 0 \\
\hline
    1 & 1 & 1\\
\hline
\end{tabular} &

\begin{tabular}{|c|c|c|}
\hline
    1 & 0 & 1 \\
\hline
    1 & 0 & 1\\
\hline
\end{tabular} 
\end{tabular}
    \caption{The nontrivial quiet patterns for the 2-color $2 \times 3$ puzzle.}
    \label{fig:qp2x3}
\end{figure}

The 2-color quiet patterns are also called parity patterns by Knuth (see \citep{KnuthProblems} for two interesting related problems and solutions) because for any square in the grid, the sum of the number in it and the numbers in the squares horizontally or vertically adjacent is even. A parity pattern is called \emph{perfect} if there is no zero row or column in it, so in Figure \ref{fig:qp2x3}, the left and middle ones are perfect and the right one is not. Using the idea of the mikado diamond introduced in \citep{Mikado}, Knuth showed \citep{Knuth} that a quiet pattern for an $M \times N$ 2-color puzzle can be ``evolved'' into a $(2M+1) \times (2N+1)$ quiet pattern. The construction essentially replaces each square in the former pattern by a $3 \times 3$ square, whose centers are now exactly 1 space apart. A former 0 is replaced by an all-zero $3\times 3$ square, and a former 1 is replaced by a `+' pattern of 1's, and then any overlapping `+' patterns are superimposed and the numbers in them added where they intersect -- see Section \ref{sect:EvoDef} for more details. The quiet patterns from Figure \ref{fig:qp2x3} evolve to the ones in Figure \ref{fig:qp5x7}, and here we see just one demonstration of a general fact that \emph{perfect} parity patterns evolve into \emph{perfect} ones; imperfect ones certainly evolve into imperfect ones.

\begin{figure}[hbtp]
    \centering
    \begin{tabular}{ccc}

\begin{tabular}{|c|c|c|c|c|c|c|}
\hline
    0 & 1 & 0 & 1 & 0 & 1 & 0 \\
\hline
    1 & 1 & 0 & 1 & 0 & 1 & 1 \\
\hline
    0 & 1 & 0 & 0 & 0 & 1 & 0 \\
\hline
    0 & 0 & 1 & 1 & 1 & 0 & 0 \\
\hline
    0 & 0 & 0 & 1 & 0 & 0 & 0 \\
\hline
\end{tabular} & 

\begin{tabular}{|c|c|c|c|c|c|c|}
\hline
    0 & 0 & 0 & 1 & 0 & 0 & 0 \\    
\hline
    0 & 0 & 1 & 1 & 1 & 0 & 0 \\
\hline
    0 & 1 & 0 & 0 & 0 & 1 & 0 \\
\hline
    1 & 1 & 0 & 1 & 0 & 1 & 1 \\
\hline
    0 & 1 & 0 & 1 & 0 & 1 & 0 \\
\hline
\end{tabular} &

\begin{tabular}{|c|c|c|c|c|c|c|}
\hline
    0 & 1 & 0 & 0 & 0 & 1 & 0 \\    
\hline
    1 & 1 & 1 & 0 & 1 & 1 & 1 \\
\hline
    0 & 0 & 0 & 0 & 0 & 0 & 0 \\
\hline
    1 & 1 & 1 & 0 & 1 & 1 & 1 \\
\hline
    0 & 1 & 0 & 0 & 0 & 1 & 0 \\
\hline
\end{tabular}
\end{tabular}
    \caption{The 1-step evolutions of the quiet patterns in Figure \ref{fig:qp2x3}}
    \label{fig:qp5x7}
\end{figure}

The results can then be evolved further and further, producing intricate designs such as the mikado diamond in \citep{Mikado} and the pattern in Figure \ref{fig:Evo5of2x3}, which shows the 5th step of evolution of the leftmost quiet pattern in Figure \ref{fig:qp2x3}: a $95 \times 127$ perfect-parity pattern. The blue in the figure represents a 1, and the white represents a 0. 

\begin{figure}[hbtp]
    \centering
    \includegraphics[width=5in]{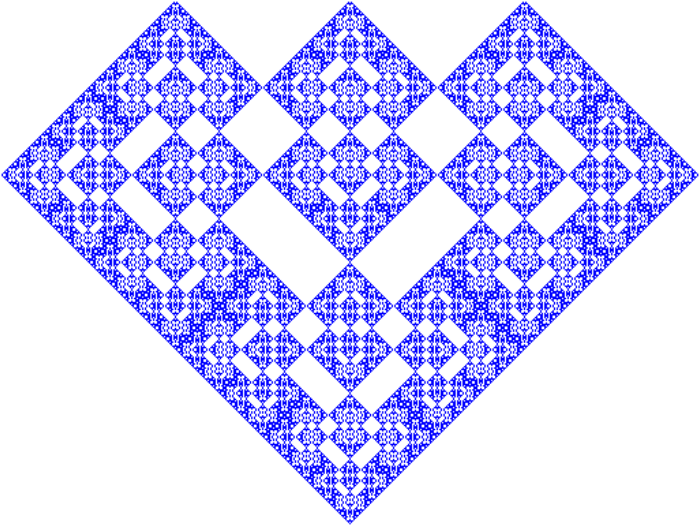}
    \caption{The pattern after 5 steps of evolution of the leftmost $2 \times 3$ quiet pattern in Figure \ref{fig:qp2x3}.}
    \label{fig:Evo5of2x3}
\end{figure}

In this paper, we generalize this construction (\cref{defn:EvoQ}) and show a way to evolve quiet patterns for any $M \times N$ puzzle with any prime number $p \ge 2$ of colors, in such a way that a ``perfect'' quiet pattern (with no zero columns or rows) evolves to another ``perfect'' quiet pattern unless $p = 2$ and there are two identical adjacent columns or rows (\cref{prop:perfect}). We show that a pattern of clicks $Q$ is a quiet pattern if and only if its evolution is a quiet pattern in \cref{thm:Q<=>Evo(Q)}. We then show that, just as in \citep{Mikado}, there is a map $\Era_p$ that acts as almost a left-inverse to the evolution map $\Evo_p$, such that
\[\Era_p \circ \Evo_p = c_p\]
where $c_p$ is the number in the center of the primitive pattern $\Evo_p(1)$ and the map $c_p$ here means ``multiplication by $c_p \pmod p$''. In particular, there is a natural left inverse for the evolution construction unless $c_p = 0$. As it turns out, there are actually infinitely many primes $p$ such that $c_p$ \textit{is} 0, but they are few and far between; namely, their density is 0 (see \cref{thm:c_p_ell_curve}). This is due to the fact that these numbers are actually the Hasse invariant of an elliptic curve $E$ which has a fractional $j$-invariant. We also show that the points on $E$ over $\F_p$ give us a group of specific quiet patterns that extend infinitely far. There is a natural action of $D_8$ on the points of $E$ using its group law, and we show in \cref{sect:EllCurves} how this action is actually reflected concretely in the quiet patterns that the points on $E$ represent.

Finally, we discuss a couple of interesting properties of the evolution patterns in \cref{sect:further}. We show in \cref{prop:Ratigan} that a quiet pattern $\pmod 2$ in an infinite grid bounded above, with a single click in the top row, is made up of Tetris ``T'' shapes and dominoes. Perhaps most significantly, in \cref{thm:5Lights}, we generalize a theorem in \citep{Mikado} that says if we click a finite number of times in a grid of squares extending infinitely far in all directions, then we will always have at least 5 lights on, and the only way to leave exactly 5 lights on is to click one of the mikado patterns. The generalization says the same is true if we replace the idea of a light being ``on'' with the state of a square simply being nonzero, as long as the number of colors is prime.

\section{Evolution of Quiet Patterns}
\label{sect:EvoDef}
In this note, we will introduce a way to ``evolve'' (or expand) a quiet pattern into a much larger quiet pattern, the expansion ratio depending on our characteristic $k$. This beautiful construction is a generalization of the mikado diamond \citep{Mikado} (see also  \citep{Knuth} Ex. \#193), and we feel it is an important part of the theory of the $k$-color Lights Out Puzzles. %For one, it will allow us to conclude that for prime $k$, $o(k(M+1)-1, k)$ is divisible by $k \cdot o(M, k)$. 
Unfortunately, the useful properties of this construction do not hold in general when $k$ is composite, and here we will restrict our attention to just prime $k$. The motivation for the name ``evolution'' is Proposition \ref{prop:infProcDown} below (and the images that follow). To experiment with evolutions, the reader is invited to visit \url{https://igor-minevich.github.io/lights-out/evolution.html}.

%Define quiet patterns and state their importance, define quiet processions, and encourage the reader to try making quiet patterns and quiet processions at our website, as well as the constructions presented further in this paper.

We first define the evolution of a single square and then that of an array of numbers, the goal being that a pattern of clicks evolves into a much larger pattern, and if the original pattern was a quiet pattern, its evolution will also be one. We will be using the following key Laurent polynomial throughout, which comes from an important idea in \citep{AlgApproach}: 
\begin{equation}
    \label{eqn:g}
    g(x, y) = x + y + 1 + x^{-1} + y^{-1}.
\end{equation}
%********Expand on this?*******
We slightly generalize the insight given in \citep{AlgApproach} for $\F_2$ to any number of colors $k$: if we allow a Laurent polynomial $f(x, y) \in \Z/k\Z[x, y, x^{-1}, y^{-1}]$ to represent the number of clicks in each square of a Lights Out Puzzle grid by letting the coefficient of $x^iy^j$ be the number of clicks in the square with coordinates $(i, j)$, then the result of those clicks is the pattern corresponding to the polynomials $f(x, y) \cdot g(x, y)$ (with some considerations of the boundary of the grid). %EXPAND ON THIS OR MAKE IT IN AN INFINITE GRID
So instead of using linear algebra and the matrix $A$ to find the result of a vector of clicks, we can simply multiply two polynomials. This brings ring theory into the picture and will allow us to arrive at some results below very easily.

\begin{defn}
\label{defn:Evo}
For a prime number $p \ge 2$, the {\bf primitive evolution pattern} $\Evo_p(1) = \Evo(1)$ is the square $(2p-1) \times (2p-1)$ pattern of clicks $\pmod p$ defined as follows. Label the coordinates of the squares with $(0, 0)$ in the very middle. The number in square $(i, j)$ is the coefficient of $x^iy^j$ in the expansion of $g(x, y)^{p-1} \pmod p$.

The primitive evolution pattern $\Evo_p(k_0)$ for a number $k_0 \in \Z/p\Z$ is defined as $k_0 \cdot \Evo_p(1)$, i.e. by the expansion of $k_0 \cdot g^{p-1} \pmod p$.
\end{defn}

\eg \label{eg:Evo_2(1)} For $p = 2$, as in \citep{Mikado}, $\Evo_2(1)$ forms a sort of ``plus'' shape, and the evolution $\Evo_2(0)$ is the zero array:
\[\Evo_2(1) = \begin{array}{ccc}
0&1&0\\
1&1&1\\
0&1&0
\end{array}, \qquad \Evo_2(0) = \begin{array}{ccc}
0&0&0\\
0&0&0\\
0&0&0
\end{array}\]

\eg \label{eg:Evo_3} For $p = 3$, the primitive evolution patterns of the two nonzero elements in $\Z/3\Z$ are:
\[\Evo_3(1) = \begin{array}{ccccc}
0 & 0 & 1 & 0 & 0 \\ 
0 & 2 & 2 & 2 & 0 \\ 
1 & 2 & 2 & 2 & 1 \\ 
0 & 2 & 2 & 2 & 0 \\ 
0 & 0 & 1 & 0 & 0
\end{array}, \qquad \Evo_3(2) = \begin{array}{ccccc}
0 & 0 & 2 & 0 & 0 \\ 
0 & 1 & 1 & 1 & 0 \\ 
2 & 1 & 1 & 1 & 2 \\ 
0 & 1 & 1 & 1 & 0 \\ 
0 & 0 & 2 & 0 & 0
\end{array}\]

\eg For $p = 5
$, we have:
\[\Evo_5(1) = \begin{array}{ccccccccc}
0 & 0 & 0 & 0 & 1 & 0 & 0 & 0 & 0 \\ 
0 & 0 & 0 & 4 & 4 & 4 & 0 & 0 & 0 \\ 
0 & 0 & 1 & 2 & 2 & 2 & 1 & 0 & 0 \\ 
0 & 4 & 2 & 1 & 0 & 1 & 2 & 4 & 0 \\ 
1 & 4 & 2 & 0 & 1 & 0 & 2 & 4 & 1 \\ 
0 & 4 & 2 & 1 & 0 & 1 & 2 & 4 & 0 \\ 
0 & 0 & 1 & 2 & 2 & 2 & 1 & 0 & 0 \\ 
0 & 0 & 0 & 4 & 4 & 4 & 0 & 0 & 0 \\ 
0 & 0 & 0 & 0 & 1 & 0 & 0 & 0 & 0
\end{array} 
, \qquad \dots,  \qquad \Evo_5(4) = \begin{array}{ccccccccc}
0 & 0 & 0 & 0 & 4 & 0 & 0 & 0 & 0 \\ 
0 & 0 & 0 & 1 & 1 & 1 & 0 & 0 & 0 \\ 
0 & 0 & 4 & 3 & 3 & 3 & 4 & 0 & 0 \\ 
0 & 1 & 3 & 4 & 0 & 4 & 3 & 1 & 0 \\ 
4 & 1 & 3 & 0 & 4 & 0 & 3 & 1 & 4 \\ 
0 & 1 & 3 & 4 & 0 & 4 & 3 & 1 & 0 \\ 
0 & 0 & 4 & 3 & 3 & 3 & 4 & 0 & 0 \\ 
0 & 0 & 0 & 1 & 1 & 1 & 0 & 0 & 0 \\ 
0 & 0 & 0 & 0 & 4 & 0 & 0 & 0 & 0
\end{array} 
 \]

\rmk The reader will note the full $D_{2 \cdot 4}$ symmetry (i.e. $90^\circ$ rotational and mirror symmetries) of the evolutions above. This is due to the fact that $g(x, y)$ is symmetric with respect to the transformations $x \lr x^{-1}, y \lr y^{-1}$, and $x \lr y$.

We generalize a fundamental property of the mikado diamond \citep{Mikado} as follows.

\begin{prop}
\label{prop:GenMikado}
Let $p$ be prime. Then the result of clicking a single $\Evo_p(k_0)$ pattern in a grid that is infinite on all sides is that all squares have a state of 0 except in five places: in the middle of the pattern and at the outermost vertices of the diamond of nonzero entries in $\Evo_p(k_0)$ (just outside of the $\Evo_p(k_0)$ pattern), and in each of those five places we get a state of precisely $k_0$. 
\end{prop}

\begin{proof}
    If $p$ is prime, then $k_0(x + y + 1 + x^{-1} + y^{-1})^p \equiv k_0(x^p + y^p + 1 + x^{-p} + y^{-p}) \pmod p$.
\end{proof}

In \cref{thm:5Lights} we will prove a sort of converse to \cref{prop:GenMikado}.

\eg As an example of Proposition \ref{prop:GenMikado}, the clicks and result of the clicks in $\Evo_3(2)$ are as follows:
\begin{center}
\begin{tabular}{ccc}
\underline{Clicks} & & \underline{Result}\\
$\begin{array}{ccccccccc}
\ddots & \vdots & \vdots & \vdots & \vdots & \vdots & \vdots & \vdots & \iddots\\ 
\cdots & 0 & 0 & 0 & 0 & 0 & 0 & 0 & \cdots \\
\cdots & 0 & 0 & 0 & 2 & 0 & 0 & 0 & \cdots \\
\cdots & 0 & 0 & 1 & 1 & 1 & 0 & 0 & \cdots \\
\cdots & 0 & 2 & 1 & 1 & 1 & 2 & 0 & \cdots \\
\cdots & 0 & 0 & 1 & 1 & 1 & 0 & 0 & \cdots \\
\cdots & 0 & 0 & 0 & 2 & 0 & 0 & 0 & \cdots \\
\cdots & 0 & 0 & 0 & 0 & 0 & 0 & 0 & \cdots \\
\iddots & \vdots & \vdots & \vdots & \vdots & \vdots & \vdots & \vdots & \ddots
\end{array}$
& $\qquad \Ra \qquad$ & $\begin{array}{ccccccccc}
\ddots & \vdots & \vdots & \vdots & \vdots & \vdots & \vdots & \vdots & \iddots\\ 
\cdots & 0 & 0 & 0 & \color{blue}{2} & 0 & 0 & 0 & \cdots \\
\cdots & 0 & 0 & 0 & 0 & 0 & 0 & 0 & \cdots \\
\cdots & 0 & 0 & 0 & 0 & 0 & 0 & 0 & \cdots \\
\cdots & \color{blue}{2} & 0 & 0 & \color{blue}{2} & 0 & 0 & \color{blue}{2} & \cdots \\
\cdots & 0 & 0 & 0 & 0 & 0 & 0 & 0 & \cdots \\
\cdots & 0 & 0 & 0 & 0 & 0 & 0 & 0 & \cdots \\
\cdots & 0 & 0 & 0 & \color{blue}{2} & 0 & 0 & 0 & \cdots \\
\iddots & \vdots & \vdots & \vdots & \vdots & \vdots & \vdots & \vdots & \ddots
\end{array}$
\end{tabular}
\end{center}

We invite the reader to use \url{https://tiny.one/lights} and, starting from the default state, click some of the patterns $\Evo_p(k_0)$ for small $p$ shown in the above examples. One can use a grid the same size as the pattern or leave an extra row and column on each side of the pattern to see that exactly 5 squares end up having the nonzero state of $k_0$. To quickly generate these patterns and see them in action, go to \url{https://igor-minevich.github.io/lights-out/evolution.html}.

Here is a curious fact arising from the construction using $g^{p-1}$ (in reality, this fact is what inspired us to use $g^{p-1}$ in the first place). If we click the middle square in a $(2p-1) \times (2p-1)$ grid $k_0$ times, record the resulting states, then start over in an empty grid and click the squares that had nonzero states the same number of times as the state itself, and repeat this process $p-2$ more times, we'll end up with precisely the states matching the primitive evolution pattern $\Evo_p(k_0)$, and if we do it again, we'll just end up with only the middle square lit, in a state of $k_0$. Formally:

\begin{thm}
\label{thm:A^kG}
Let $A$ be the activation matrix for the standard Lights Out Puzzle on the $(2p-1)\times(2p-1)$ grid. Let $\vec v$ be the $(2p-1)^2$-dimensional vector corresponding to a $(2p-1) \times (2p-1)$ grid consisting of 0's everywhere except the middle square, which has $k_0$ in it. Then $A^{p-1}\vec v$ generates the $\Evo_p(k_0)$ pattern, and $A^p\vec v = \vec v$. 
\end{thm}
\begin{proof}
Applying $A$ to a vector of clicks to get the vector of states is equivalent to multiplying $g(x, y)$ by the polynomial representing those clicks to get the ``state'' polynomial.
\end{proof}

We will use the number of clicks in the center of the primitive $\Evo_p(1)$ construction through much of the rest of the paper, and we will find several equivalent ways to arrive at these numbers. For one, this is the Hasse invariant of an elliptic curve. Thus, we make the following definition.
\begin{defn}
\label{defn:c_p}
    $c_p$ is the number of clicks in the center of the primitive evolution pattern $\Evo_p(1)$, i.e. $c_p$ is the constant coefficient of $(x + y + 1 + x^{-1} + y^{-1})^{p-1} \pmod p$.
\end{defn}

\begin{cor}
The number of clicks $c_p$ in the center of $\Evo_p(1)$ is given by taking the very middle entry (middle entry in the middle column) in the $(2p-1)^2 \times (2p-1)^2$ matrix $A^{p-1}$, where $A$ is the activation matrix for the $(2p-1) \times (2p-1)$ Lights Out Puzzle.
\end{cor}
\begin{proof}
This follows from basic linear algebra. Each column of the matrix $A$ represents the result of a click in a square, which adds 1 to the neighboring squares and itself. Now letting $m = 2p^2 - 2p + 1$, the vector $\vec e_m$ is the vector with a single 1 in the middle entry of the $(2p-1)^2$-length vector and 0's everywhere else. Then $A^{p-1}\vec e_m$ is precisely the primitive evolution pattern $\Evo_p(1)$; this corresponds to the middle column of $A^{p-1}$, and $c_p$ is the middle ($m$-th) entry in that column.
\end{proof}

\begin{defn}
\label{defn:EvoQ}
The {\bf evolution} $\Evo_p(Q) = \Evo(Q)$ {\bf of an $M \times N$ array $Q$} of clicks $\pmod p$ is defined as follows. Write $Q$ as the polynomial $f(x, y) = \sum_{i=0}^{N-1} \sum_{j=0}^{M-1} c_{ij}x^iy^j$. Then $\Evo_p(Q)$ is the $(p(M+1)-1) \times (p(N+1)-1)$ array of clicks corresponding to the polynomial $f(x^p, y^p)g(x, y)^{p-1} \pmod p$. 
%Could multiply by x^{p-1}y^{p-1} to make all coordinates nonzero, but then that contradicts the definition of Evo(1).

We will denote by $\Evo_p^{(n)}(Q) = \Evo^{(n)}(Q)$ the $n$-fold evolution of the quiet pattern, i.e. the $n$-th iteration of the $\Evo_p$ function.
\end{defn}

\rmk \label{rmk:visualizingEvo} A more visual (and perhaps intuitive) way to look at the evolution of an $M \times N$ array $Q$ is the following: $\Evo_p(Q)$ is the pattern defined by replacing each cell with coordinates $(i, j)$ ($1 \le i \le M, 1 \le j \le N$) and the number $k_0$ by the $\Evo_p(k_0)$ pattern centered at $(pi, pj)$ and adding the cells from the various evolutions $\pmod p$ where they overlap; see Figure \ref{fig:evo_3-diamonds}. This is because for a pattern given by $f(x, y)$, for each $i = 0, \dots, N-1$ and $j = 0, \dots, M-1$ the evolution construction takes the primitive evolution pattern given by $g(x, y)^{p-1}$ and places it centered at $(pi, pj)$ and multiplied by the coefficient of $x^iy^j$ in $f(x, y)$ (this comes from $f(x^p, y^p)$). Note that the squares which have nonzero clicks always form a diamond, only two diamonds ever overlap, and there are holes in between the outermost diamonds, where we are left with 0's. 

Figure \ref{fig:Evo_3^6(1)} continues to show further evolutions of the patterns in Figure \ref{fig:evo_3-diamonds}, replacing the light blue color in Figure \ref{fig:evo_3-diamonds} with a darker blue to show better contrast between the colors. Figure \ref{fig:evo_5^4(1)} shows evolutions of the same $1 \times 1$ array (1) with $k = 5$ colors. The fractal nature of the patterns is evident from these pictures.

\begin{figure}[hbtp]
    \centering
    \includegraphics[height=3in]{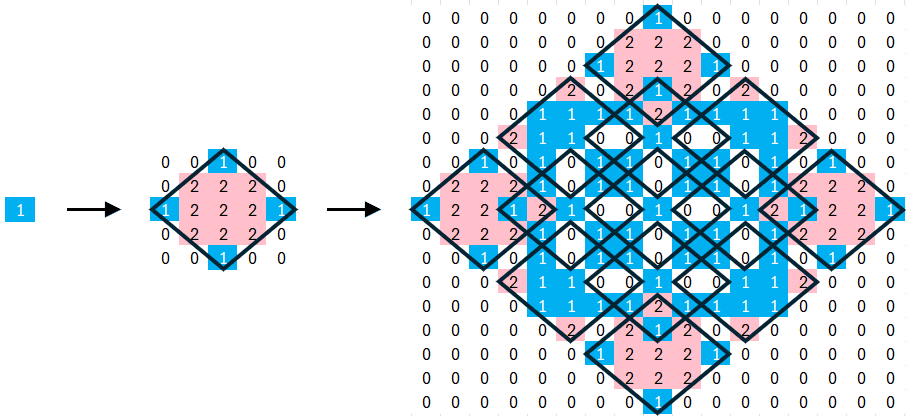}
    \caption{Two steps of the evolution of the $1 \times 1$ array (1) for $k = 3$, showing how the diamonds coming from $\Evo_p(1)$ are superimposed in $\Evo_p^{(2)}(1)$.}
    \label{fig:evo_3-diamonds}
\end{figure}

%  \quad \to 
% \begin{array}{ccccc}
% 0 & 0 & 1 & 0 & 0 \\
% 0 & 2 & 2 & 2 & 0 \\
% 1 & 2 & 2 & 2 & 1 \\
% 0 & 2 & 2 & 2 & 0 \\
% 0 & 0 & 1 & 0 & 0 \\
% \end{array} \quad \to 
% \begin{array}{ccccccccccccccccc}
% 0 & 0 & 0 & 0 & 0 & 0 & 0 & 0 & 1 & 0 & 0 & 0 & 0 & 0 & 0 & 0 & 0 \\
% 0 & 0 & 0 & 0 & 0 & 0 & 0 & 2 & 2 & 2 & 0 & 0 & 0 & 0 & 0 & 0 & 0 \\
% 0 & 0 & 0 & 0 & 0 & 0 & 1 & 2 & 2 & 2 & 1 & 0 & 0 & 0 & 0 & 0 & 0 \\
% 0 & 0 & 0 & 0 & 0 & 2 & 0 & 2 & 1 & 2 & 0 & 2 & 0 & 0 & 0 & 0 & 0 \\
% 0 & 0 & 0 & 0 & 1 & 1 & 1 & 1 & 2 & 1 & 1 & 1 & 1 & 0 & 0 & 0 & 0 \\
% 0 & 0 & 0 & 2 & 1 & 1 & 0 & 0 & 1 & 0 & 0 & 1 & 1 & 2 & 0 & 0 & 0 \\
% 0 & 0 & 1 & 0 & 1 & 0 & 1 & 1 & 0 & 1 & 1 & 0 & 1 & 0 & 1 & 0 & 0 \\
% 0 & 2 & 2 & 2 & 1 & 0 & 1 & 1 & 0 & 1 & 1 & 0 & 1 & 2 & 2 & 2 & 0 \\
% 1 & 2 & 2 & 1 & 2 & 1 & 0 & 0 & 1 & 0 & 0 & 1 & 2 & 1 & 2 & 2 & 1 \\
% 0 & 2 & 2 & 2 & 1 & 0 & 1 & 1 & 0 & 1 & 1 & 0 & 1 & 2 & 2 & 2 & 0 \\
% 0 & 0 & 1 & 0 & 1 & 0 & 1 & 1 & 0 & 1 & 1 & 0 & 1 & 0 & 1 & 0 & 0 \\
% 0 & 0 & 0 & 2 & 1 & 1 & 0 & 0 & 1 & 0 & 0 & 1 & 1 & 2 & 0 & 0 & 0 \\
% 0 & 0 & 0 & 0 & 1 & 1 & 1 & 1 & 2 & 1 & 1 & 1 & 1 & 0 & 0 & 0 & 0 \\
% 0 & 0 & 0 & 0 & 0 & 2 & 0 & 2 & 1 & 2 & 0 & 2 & 0 & 0 & 0 & 0 & 0 \\
% 0 & 0 & 0 & 0 & 0 & 0 & 1 & 2 & 2 & 2 & 1 & 0 & 0 & 0 & 0 & 0 & 0 \\
% 0 & 0 & 0 & 0 & 0 & 0 & 0 & 2 & 2 & 2 & 0 & 0 & 0 & 0 & 0 & 0 & 0 \\
% 0 & 0 & 0 & 0 & 0 & 0 & 0 & 0 & 1 & 0 & 0 & 0 & 0 & 0 & 0 & 0 & 0 \\
% \end{array}\]

\newcommand{\hsone}{\hspace{0.1in}}
\begin{figure}[hbtp]
    \centering
    \includegraphics[width=0.02\linewidth]{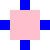} \hsone
    \includegraphics[width=0.05\linewidth]{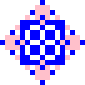} \hsone
    \includegraphics[width=0.1\linewidth]{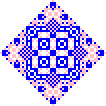} \hsone
    \includegraphics[width=0.15\linewidth]{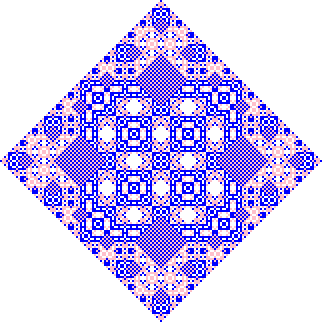} \hsone
    \includegraphics[width=0.2\linewidth]{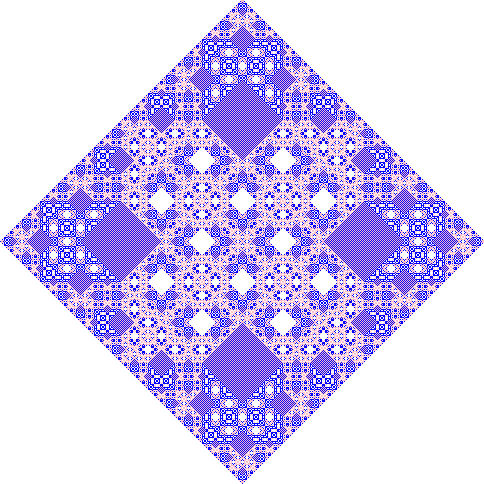} \hsone
    \includegraphics[width=0.25\linewidth]{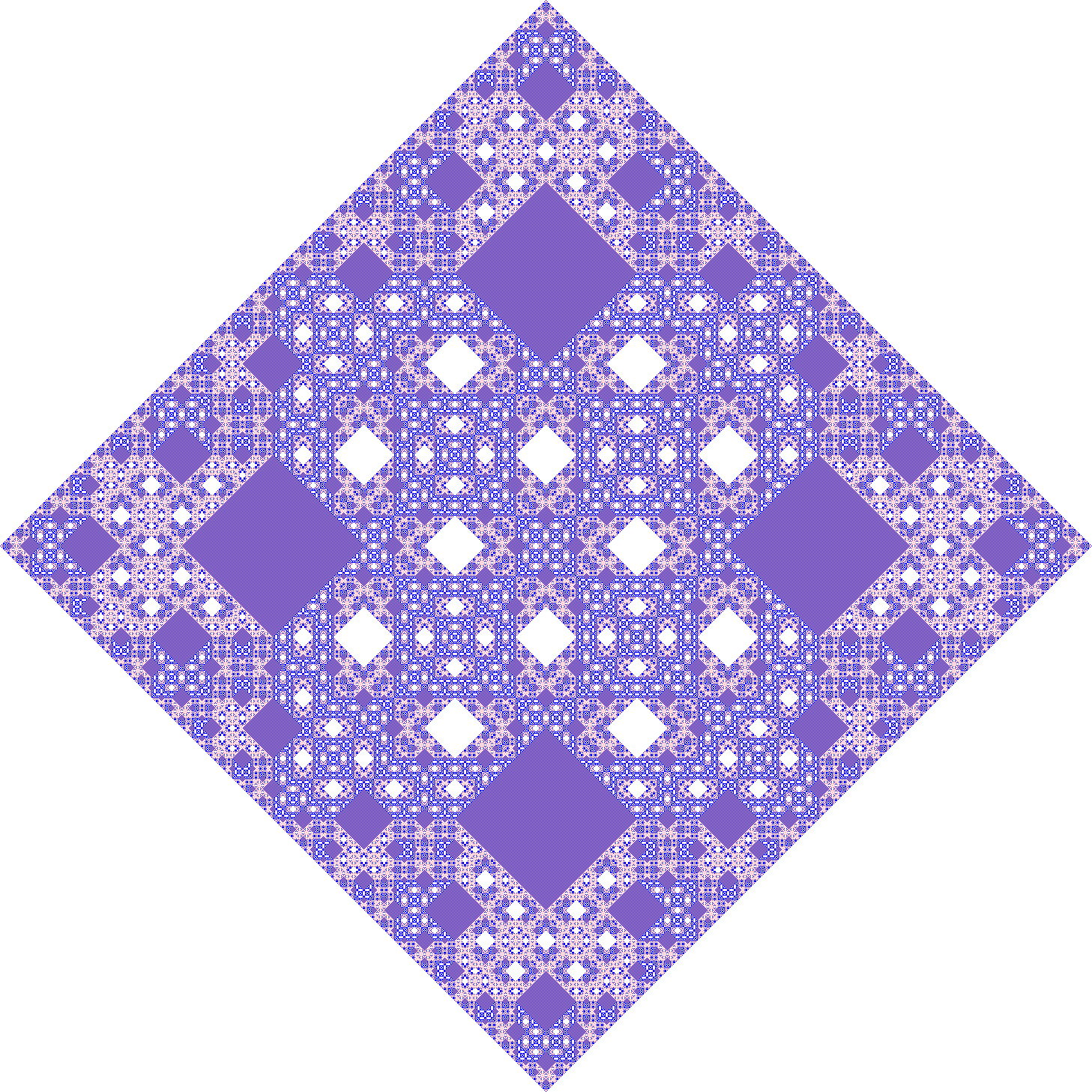}
    \includegraphics[width=0.05\linewidth]{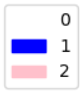}
    \caption{The six 3-color evolutions $\Evo_3(1), \Evo_3^{(2)}(1), \dots, \Evo_3^{(6)}(1)$ (left to right). Note: images are not to scale.}
    \label{fig:Evo_3^6(1)}
\end{figure}

\begin{figure}[hbtp]
    \centering
    \includegraphics[width=0.05\linewidth]{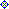} \hsone
    \includegraphics[width=0.1\linewidth]{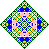} \hsone
    \includegraphics[width=0.2\linewidth]{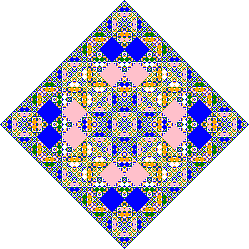} \hsone
    \includegraphics[width=0.4\linewidth]{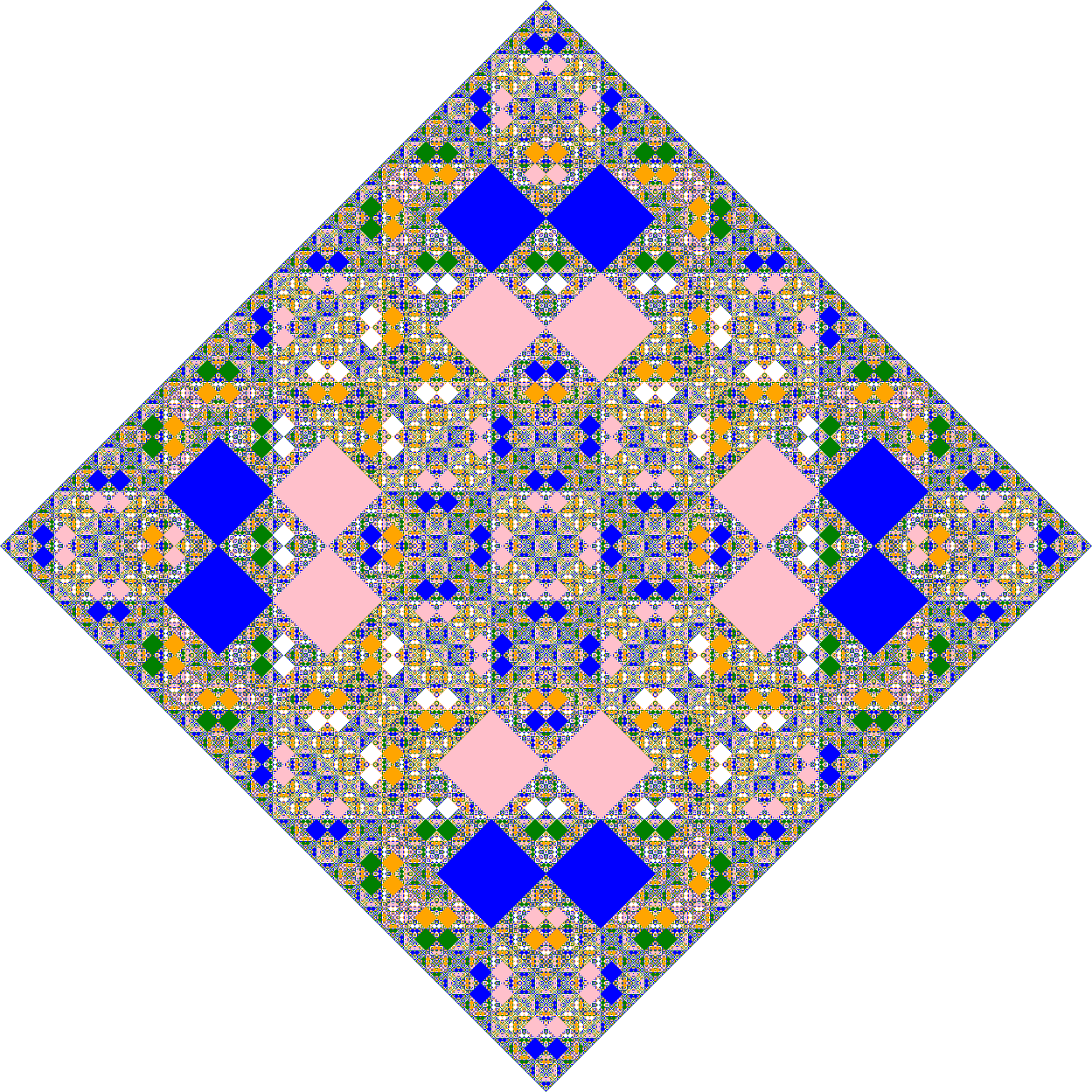}
    \includegraphics[width=0.05\linewidth]{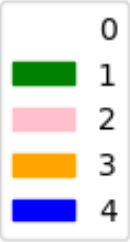}
    \caption{The four 5-color evolutions $\Evo_5(1), \Evo_5^{(2)}(1), \Evo_5^{(3)}(1), \Evo_5^{(4)}(1)$(left to right). See legend for description of colors. Note: images are not to scale.}
    \label{fig:evo_5^4(1)}
\end{figure}

%\FloatBarrier
\rmk Note that since $p$ is prime, the pattern $\Evo_p^{(n)}(1)$ is simply given by the polynomial $g(x, y)^{p^n -1}$: $\Evo_p^{(2)}(1)$ is given by $g(x^p, y^p)^{p-1}g(x, y)^{p-1} = g(x, y)^{p(p-1)}g(x, y)^{p-1} = g(x, y)^{p^2-1}$, and an easy induction shows this holds in general.
%THIS DOES NOT HOLD FOR COMPOSITE k!!

\eg \begin{samepage}
    Let $p = 3$ and $Q$ be the $3 \times 4$ quiet pattern below. Its evolution is an $11 \times 14$ quiet pattern: 
\[Q = \begin{array}{cccc}
1 & 2 & 2 & 1\\
0 & 1 & 1 & 0\\
1 & 2 & 2 & 1
\end{array} \quad \Ra \quad \Evo_3(Q) = \begin{array}{cccccccccccccc}
0 & 0 & 1 & 0 & 0 & 2 & 0 & 0 & 2 & 0 & 0 & 1 & 0 & 0\\
0 & 2 & 2 & 2 & 1 & 1 & 1 & 1 & 1 & 1 & 2 & 2 & 2 & 0\\
1 & 2 & 2 & 1 & 2 & 1 & 0 & 0 & 1 & 2 & 1 & 2 & 2 & 1\\
0 & 2 & 2 & 2 & 1 & 2 & 1 & 1 & 2 & 1 & 2 & 2 & 2 & 0\\
0 & 0 & 1 & 0 & 2 & 1 & 2 & 2 & 1 & 2 & 0 & 1 & 0 & 0\\
0 & 0 & 0 & 1 & 2 & 2 & 0 & 0 & 2 & 2 & 1 & 0 & 0 & 0\\
0 & 0 & 1 & 0 & 2 & 1 & 2 & 2 & 1 & 2 & 0 & 1 & 0 & 0\\
0 & 2 & 2 & 2 & 1 & 2 & 1 & 1 & 2 & 1 & 2 & 2 & 2 & 0\\
1 & 2 & 2 & 1 & 2 & 1 & 0 & 0 & 1 & 2 & 1 & 2 & 2 & 1\\
0 & 2 & 2 & 2 & 1 & 1 & 1 & 1 & 1 & 1 & 2 & 2 & 2 & 0\\
0 & 0 & 1 & 0 & 0 & 2 & 0 & 0 & 2 & 0 & 0 & 1 & 0 & 0\\
\end{array}.\]
\end{samepage}

Indeed, this is not a coincidence. An $M \times N$ quiet pattern $Q$ is an array that corresponds to a polynomial $f(x, y)$ such that $s(x, y) = f(x,y)g(x,y)$, the ``state'' polynomial that corresponds to finding the state of the grid after the clicks in $Q$ (starting from the all-zero state), only has nonzero terms where the powers of $x$ are -1 or $N$ and powers of $y$ are -1 or $M$. For such an array, since $p$ is prime, the result of clicking its evolution is 
\begin{equation}
    \label{eqn:EvoQ}
    f(x^p, y^p)g(x, y)^{p-1}\cdot g(x, y) = f(x^p, y^p) g(x, y)^p = f(x^p, y^p)g(x^p, y^p) = s(x^p, y^p),
\end{equation}
which is the same state polynomial except with its terms spread out $p$ times. Thus, when $Q$ is a quiet pattern, $s$ only has nonzero terms where the powers of $x$ are $-p$ or $Np$ and powers of $y$ are $-p$ or $Mp$, so $f(x^p, y^p)g(x, y)^{p-1}$ forms a quiet $(Mp+p-1) \times (Np+p-1)$ pattern. 

More generally, the resulting state (starting from all-zero) after clicking $\Evo_p(Q)$ has 0's everywhere except at the centers of the primitive evolution patterns that replace the single squares in $Q$ (see Remark \ref{rmk:visualizingEvo}), and the resulting states in those centers are exactly the resulting states of the corresponding entries in $Q$. Let us state this more precisely. The following theorem can be proven using \eqref{eqn:EvoQ}, but we will also give a more visual / conceptual proof below.

\begin{thm}
\label{thm:Q<=>Evo(Q)}
Let $p$ be prime. The evolution $\Evo_p(Q)$ is a quiet pattern if and only if $Q$ is a quiet pattern. 

In fact, if we number the rows of $Q$ as $1, 2, \dots, m$ and columns $1, 2, \dots, n$, and number the rows of $\Evo_p(Q)$ as $1, 2, \dots, M = p(m+1) - 1$ and columns $1, 2, \dots, N = p(n+1)-1$, then the resulting change in state of square $(i, j)$ after clicking $Q$ is the same as that of square $(pi, pj)$ after clicking $\Evo_p(Q)$ (for all $i = 1, 2, \dots, m$, $j = 1, 2, \dots, n$). All other states after clicking $\Evo_p(Q)$ remain unchanged.
\end{thm}

\eg Say $p = 3$ and the configuration of clicks and resulting states from $Q$ is as follows:
\begin{center}
\begin{tabular}{ccc}
\underline{Clicks} & & \underline{Result}\\
$\begin{array}{ccc}
1 & 0 & 2\\
2 & 1 & 2
\end{array}$
& $\qquad \Ra \qquad$ & $\begin{array}{ccc}
0 & 1 & 0\\
1 & 2 & 2
\end{array}$
\end{tabular}
\end{center}
Then the corresponding clicks and result after the evolution are:
\begin{center}
\begin{tabular}{ccc}
\underline{Clicks} & & \underline{Result}\\
$\begin{array}{ccccccccccc}
0 & 0 & 1 & 0 & 0 & 0 & 0 & 0 & 2 & 0 & 0\\
0 & 2 & 2 & 2 & 0 & 0 & 0 & 1 & 1 & 1 & 0\\
1 & 2 & 2 & 2 & 1 & 0 & 2 & 1 & 1 & 1 & 2\\
0 & 2 & 1 & 2 & 0 & 1 & 0 & 1 & 0 & 1 & 0\\
0 & 1 & 2 & 1 & 2 & 2 & 2 & 1 & 0 & 1 & 0\\
2 & 1 & 1 & 2 & 1 & 2 & 1 & 2 & 1 & 1 & 2\\
0 & 1 & 1 & 1 & 2 & 2 & 2 & 1 & 1 & 1 & 0\\
0 & 0 & 2 & 0 & 0 & 1 & 0 & 0 & 2 & 0 & 0\\
\end{array}$
& $\qquad \Ra \qquad$ & $\begin{array}{ccccccccccc}
0 & 0 & 0 & 0 & 0 & 0 & 0 & 0 & 0 & 0 & 0\\
0 & 0 & 0 & 0 & 0 & 0 & 0 & 0 & 0 & 0 & 0\\
0 & 0 & \color{blue}{\mathbf{0}} & 0 & 0 & \color{blue}{\mathbf{1}} & 0 & 0 & \color{blue}{\mathbf{0}} & 0 & 0\\
0 & 0 & 0 & 0 & 0 & 0 & 0 & 0 & 0 & 0 & 0\\
0 & 0 & 0 & 0 & 0 & 0 & 0 & 0 & 0 & 0 & 0\\
0 & 0 & \color{blue}{\mathbf{1}} & 0 & 0 & \color{blue}{\mathbf{2}} & 0 & 0 & \color{blue}{\mathbf{2}} & 0 & 0\\
0 & 0 & 0 & 0 & 0 & 0 & 0 & 0 & 0 & 0 & 0\\
0 & 0 & 0 & 0 & 0 & 0 & 0 & 0 & 0 & 0 & 0\\
\end{array}$
\end{tabular}
\end{center}

\begin{proof}[Proof of Theorem \ref{thm:Q<=>Evo(Q)}]
Note that the second paragraph of Theorem \ref{thm:Q<=>Evo(Q)} implies the first. By Proposition \ref{prop:GenMikado}, the resulting clicks in $\Evo(Q)$ produce a net of 0 clicks everywhere except at the centers of the primitive evolution patterns replacing the single squares, and the result at a given center of an evolution is precisely the sum of the clicks in the corresponding square from $Q$ and the clicks in the squares adjacent to it in $Q$, since the vertex of a diamond corresponding to a primitive evolution pattern lands either precisely at the edge of the grid or at the center of the adjacent primitive evolution pattern. Thus, the resulting states of the centers in $\Evo(Q)$ are precisely the same as the resulting states of the squares in $Q$.
\end{proof}

\rmk The statement of Proposition \ref{prop:GenMikado} holds true as we keep evolving the resulting pattern into the next, since the $n$-th evolution $\Evo^{(n)}(k_0)$ corresponds to the polynomial $k_0(x + y + 1 + x^{-1} + y^{-1})^{p^n} \equiv k_0(x^{p^n} + y^{p^n} + 1 + x^{-p^n} + y^{-p^n}) \pmod p$ for prime $p$. So as we evolve the pattern further and further, the five squares that end up with nontrivial states as a result of the clicks drift $p$ times farther apart each time (the middle staying put and the other four fleeing faster and faster).

%*****QUESTION*****: is there a better way to define the inverse map? For a pattern of clicks $f(x, y)$, as the $p$-th root of the state polynomial $f(x, y)g(x, y)$ if it exists?

\begin{thm}
\label{thm:Era_p}
Let $p$ be a prime number. For any $M, N \equiv -1 \pmod p$ (with $M, N > p$), there is a map $\Era_p$ that takes $M \times N$ quiet patterns to $(M+1-p)/p \times (N+1-p)/p$ quiet patterns. The map simply erases all rows and columns whose indices are not multiples of $p$ (where rows and columns are 1-indexed).
\end{thm}
\begin{proof}
Let $x_{ij}$ be the clicks in the square in position $(i, j)$, and suppose for the moment that we had a quiet pattern in an infinite grid. Note that by the definition of a quiet pattern we have
\begin{equation}
    \label{eqn:qp_rel}
    x_{i, j} + x_{i-1, j} + x_{i+1, j} + x_{i, j-1} + x_{i, j+1} \equiv 0 \pmod{p}.
\end{equation}
We will compactify the notation slightly by writing $x_{i\pm 1, j}$, e.g. to mean $x_{i+1, j} + x_{i-1, j}$, turning Equation \eqref{eqn:qp_rel} into 
\begin{equation}
\label{eqn:qp_rel_short}
x_{i, j} + x_{i\pm 1, j} + x_{i, j\pm1} \equiv 0 \pmod p.    
\end{equation}
Now, apply \eqref{eqn:qp_rel} again centered at each of the squares corresponding to the entries in \eqref{eqn:qp_rel_short} to find
\begin{equation}
    \label{eqn:qp_rel2}
    x_{i\pm 2, j} + x_{i, j \pm 2} + 2(x_{i \pm 1, j - 1} + x_{i \pm 1, j+1} + x_{i\pm 1, j} + x_{i, j \pm 1} + x_{i, j}) \equiv 0 \pmod{p},
\end{equation}
Apply \eqref{eqn:qp_rel} again, centered at each of the squares corresponding to the entries in \eqref{eqn:qp_rel2}, multiplied by the coefficients in \eqref{eqn:qp_rel2}. If we repeat this process a total of $p$ times, the result will be:
\begin{equation}
    \label{eqn:qp_rel_p}
    x_{i, j} + x_{i-p, j} + x_{i+p, j} + x_{i, j-p} + x_{i, j+p} \equiv 0 \pmod{p}.
\end{equation}
This is because each application of \eqref{eqn:qp_rel} to the previous equation is equivalent to multiplying the activation matrix for the puzzle by the vector of constants specified in the equation, starting with a solo square in the middle. This is only being done on finitely many terms, so we can think of this as finite-dimensional linear algebra. Applying the activation matrix $p$ times produces $A^p$, the result of which (as seen in the proof of Theorem \ref{thm:A^kG}) consists precisely of the square at the center and the squares horizontally and vertically $p$ away from the center.

But here, we do not have an infinite grid; instead we have a finite $M \times N$ grid, so we can start at each place $(i, j) = (pa, pb)$ with $a = 1, 2, \dots, (M+1-p)/p$ and $b = 1, 2, \dots, (N+1-p)/p$, and note that all of the terms that come into play lie in the $M \times N$ grid up until the very last application, at which point we may have terms that are just outside of the $M \times N$ grid; this amounts to treating those terms as 0's. In total, this implies that the relations \eqref{eqn:qp_rel_p} hold for the squares of the form $(pa, pb)$ just mentioned, so if we erase all the squares but these, we will still have a quiet pattern. 
\end{proof}

Note that if we try to erase a different set of rows and columns, or if we start out with a size $M \times N$ where $M$ or $N$ are not $-1 \pmod p$ then at some point we will end up with a relation that contains terms outside of the grid and we will need to apply \eqref{eqn:qp_rel} for those terms; since these terms are set to 0, they cannot affect other squares and the relationship will no longer hold.

\begin{thm}
\label{thm:EraEvo=cp}
Let $k = p$ be a prime number. Applying evolution $\Evo_p$ to a pattern of clicks and then applying $\Era_p$ is equivalent to multiplying the original pattern by $c_p$. In symbols:
\[\Era_p \circ \Evo_p = c_p\]
\end{thm}
\begin{proof}
Let $Q$ be the starting $m \times n$ grid, with clicks $x_{ij}$ in positions $(i, j)$. By definition, $\Evo_p(Q)$ is an $M \times N$ grid that is a superposition of the $\Evo_p(x_{ij})$ patterns, which have $x_{ij} \cdot c_p$ in the middle, with $M = p(m+1) - 1$ and $N = p(n+1) - 1$. Only the $\Evo_p(x_{ij})$ pattern contributes to the central square (the patterns do not overlap there), and the erasing procedure $\Era_p$ leaves precisely those centers, so the new grid, which is precisely $m \times n$ again, has precisely the numbers $x_{ij} \cdot c_p$ in the $(i, j)$ position.
\end{proof}

We can easily write down a combinatorial expression for $c_p$ as the constant coefficient of $(x + x^{-1} + 1 + y + y^{-1})^{p-1}$: if we take $b$ $x$'s, $c$ 1's, and $d$ $y$'s, then we must also take $b$ $x^{-1}$ terms and $d$ $y^{-1}$ terms, so we have:
\begin{equation}
    \label{c_k_comb}
    c_k = \sum_{2b + c + 2d = k-1} \frac{(k-1)!}{(b!)^2c!(d!)^2}
\end{equation}
In particular, for odd $p$, $c$ must be even, so replacing $c$ by $2c$ above, we have
\begin{equation}
    c_p = \sum_{c=0}^{(p-1)/2} \sum_{b=0}^{(p-1)/2-c} \frac{(p-1)!}{(b!)^2(2c)!\left(\left(\frac{p-1}2 - b - c\right)!\right)^2}
\end{equation}

\rmk \label{c_p^n} The number of clicks $c_p^{(n)}$ in the center of $\Evo_p^{(n)}(1)$ is just $c_p^n$, since in evolving from $\Evo_p^{(n-1)}(1)$, the middle square is replaced by the primitive pattern $\Evo_p(c_p^{(n-1)})$. The evolutions of the other squares do not effect the middle, so the middle just becomes $c_p^{(n-1)}\cdot c_p$.

\eg Since $c_2^{(n)} = 1$ for all $n$ ($c_2 = 1$), replacing $k$ by $2^n$ in \eqref{c_k_comb}, the number
\[\sum_{2b + c + 2d = 2^n-1} \frac{(2^n-1)!}{(b!)^2c!(d!)^2}\]
is always odd.

\eg Similarly, since $c_3 = 2$ and $c_5 = 1$, we can say 
\[\sum_{2b + c + 2d = 3^n-1} \frac{(3^n-1)!}{(b!)^2c!(d!)^2} \equiv (-1)^n \pmod 3 \qquad \text{and} \qquad \sum_{2b + c + 2d = 5^n-1} \frac{(5^n-1)!}{(b!)^2c!(d!)^2} \equiv 1 \pmod 5.\]

Table \ref{tab:c_p} shows the numbers $c_p$ for small primes $p$; these are the OEIS sequence A275745 (see Remark \ref{rmk:OEIS}):

\begin{table}[htbp]
    \centering
    \begin{tabular}{|c|c|c|c|c|c|c|c|c|c|c|c|c|c|c|c|c|c|c|c|c|c|}
        \hline
        $p$ & 2 & 3 & 5 & 7 & 11 & 13 & 17 & 19 & 23 & 29 & 31 & 37 & 41 & 43 & 47 & 53 & 59 & 61 & 67 & 71 & 73 \\
        \hline
        $c_p$ & 1 & 2 & 1 & 0 & 7 & 11 & 2 & 4 & 0 & 27 & 0 & 27 & 10 & 4 & 8 & 43 & 55 & 59 & 12 & 63 & 10\\
        \hline
    \end{tabular}
    \caption{The numbers $c_p$ of clicks in the center of $\Evo_p(1)$ for small primes $p$.}
    \label{tab:c_p}
\end{table}

The primes for which $c_p = 0$ are: 7, 23, 31, 79, 167, 431, 479, 983, 1303, $\dots$ (this is OEIS A275777).

\section{Elliptic Curves}
\label{sect:EllCurves}
To learn more about these numbers $c_p$, we first clear the denominator of $g(x, y) = 0$:
\begin{equation}
    \label{eqn:g_den_cleared}
    E: x^2y + xy^2 + xy + y + x = 0
\end{equation}
and projectivize it:
\begin{equation}
    \label{eqn:g_proj}
    \tilde E: f(x, y, z) = x^2y + xy^2 + xyz + yz^2 + xz^2 = 0
\end{equation}
% Let $(x, y)$ be a point satisfying $g(x, y) = 0$. Then this point is naturally identified with $(x, y, 1)$ on the projectivization $\tilde E$ from \cref{eqn:g_proj}.  
This is an elliptic curve with 4 rational points, 3 of which are on the line at infinity $z = 0$, and these form a subgroup isomorphic to $\Z/4\Z$ (see \cref{tab:triv_gp_a=1}). None of these actually give solutions to $g(x, y) = 0$ so we will call these four the \textit{trivial points} on $\tilde E$. We take $(-1, 1, 0)$ to be the identity on $\tilde E$. $E$ is isomorphic to the following Tate normal form for a point of order 4:
\begin{equation}
    \label{eqn:E_4(1)}
    E_4(1): y^2 + xy + y = x^3 + x^2
\end{equation}
with projectivization
\[\tilde E_4(1): y^2z + xyz + yz^2 = x^3 + x^2z.\]
Table \ref{tab:triv_gp_a=1} summarizes the points in the trivial group of order 4 on $\tilde E$ and $\tilde E_4(1)$:
\begin{table}[hbtp]
    \centering
    \begin{tabular}{ccc}
        Point on $\tilde E$ & Point on $\tilde E_4(1)$ & Order\\
        (1, 0, 0) & (0, -1, 1) & 4\\
        (0, 1, 0) & (0, 0, 1) & 4\\
        (0, 0, 1) & (-1, 0, 1) & 2 \\
        (-1, 1, 0) & (0, 1, 0) & 1\\
    \end{tabular}
    \caption{The trivial subgroup of order 4 on $\tilde E$ and $\tilde E_4(1)$}
    \label{tab:triv_gp_a=1}
\end{table}

The isomorphism $\tilde E \xr{\sim} \tilde E_4(1)$ that takes $(-1, 1, 0)$ on $\tilde E$ to $(0, 1, 0)$ on $\tilde E_4(1)$ is given by 
\[\phi: \tilde E \to \tilde E_4(1), \quad (x, y, z) \mapsto (-z, -x, x + y + z)\]
with inverse given by 
\[\psi: \tilde E_4(1) \to \tilde E, \quad (u, v, w) \mapsto (-v, u + v + w, -u).\]

Now we can use the theory of elliptic curves to understand the numbers $c_p$.
\begin{thm}
\label{thm:c_p_ell_curve}
The number $c_p$ is the Hasse invariant of the elliptic curve $E$ in \cref{eqn:g_den_cleared}. There are infinitely many primes $p$ for which $c_p = 0$, but the density of such primes is 0.
\end{thm}
\begin{proof}
According to \citep[Prop. IV.4.21, Cor. IV.4.22]{Hartshorne}, the coefficient of $(xyz)^{p-1}$ in $f(x, y, z)^{p-1}$ is precisely the Hasse invariant of $\tilde E$. Since this is the same as the constant coefficient in $g(x, y)^{p-1}$, $c_p$ is precisely the Hasse invariant. A quick SageMath computation shows that $E_4(1)$ has $j$-invariant equal to -1/15 \citep[Theorem II.6.1]{silverman2013advanced}, which means it does not have complex multiplication. Elkies' theorem \citep{elkies1987existence} says that there are infinitely many primes at which any elliptic curve is supersingular; however, the actual set of supersingular primes has density 0 \citep{silverman2009arithmetic, elkies1992distribution}.
\end{proof}
\rmk\label{rmk:OEIS} The Online Encyclopedia of Integer Sequences (OEIS) has at least three sequences directly relating to the $\Evo$ construction, specifically the sequences related to the elliptic curve
\[E': y^2 + xy + y = x^3 + x^2 - 10x - 10 \pmod p.\]
Another quick Sagemath computation shows $E_4(1)$ and $E'$ are isogenous. In fact, the isogeny from the Tate normal form $E_4(1)$ to $E'$ is given by taking the kernel to be the group of 4 trivial points shown in Table \ref{tab:triv_gp_a=1}. This means the two curves have the same Hasse invariants, so the sequence $\{c_p\}$ giving the number of clicks in the center of $\Evo_p(1)$ is congruent to OEIS sequence A275745 $\pmod p$, which gives the Hasse invariant for $E'$. The number of points on $E'$ and $E_4(1) \pmod p$ is given by A275742. Note that the number of points on $E \pmod p$ is two less than A275742 because the curves $E'$ and $E_4(1)$ are in Weierstrass form and have only one point at infinity, whereas three of the points in Table \ref{tab:triv_gp_a=1} are infinite on $E$. Finally, the primes $p$ such that the primitive evolution pattern $\Evo_p(1)$ has a 0 in the middle, making $\Era_p \circ \Evo_p = 0$, are the sequence A275777 of primes $p$ such that there are exactly $p$ finite points on $E'$ over $\F_p$.
%The isogeny is explicitly given by ((x^4 + x^3 + 2*x^2 + 2*x + 1)/(x^3 + x^2), (x^5*y + 2*x^4*y - 2*x^4 - x^3*y - 5*x^3 - 4*x^2*y - 6*x^2 - 5*x*y - 4*x - 2*y - 1)/(x^5 + 2*x^4 + x^3))

Now, let us explore how this elliptic curve relates to quiet patterns in the Lights Out Puzzle. Consider an infinite grid $G$, stretching infinitely far in all four directions from a fixed starting point with coordinates $(0, 0)$. Let $p$ be the number of colors in the puzzle, as before.
\begin{defn}
A quiet pattern in $G$ is a function $c: \Z \times \Z \to \Z/p\Z$, such that $c(i, j) + c(i-1, j) + c(i+1, j) + c(i, j-1) + c(i, j+1) \equiv 0 \pmod p$ for all $i, j \in \Z$.
\end{defn}

We can easily make infinitely (in fact, uncountably) many quiet patterns in $G$ by setting $c(0, j)$ and $c(1, j)$ in an arbitrary way for all $j \in \Z$ and then proceeding from there to ``chase the lights'' to the left and right of these two columns. To be explicit, by chasing the lights to the right, we mean that we recursively set $c(i, j)$ to be the unique value so that the state of square $(i-1, j)$ is 0, i.e. $c(i, j) = -c(i-1, j) - c(i-1, j+1) - c(i-1, j-1) - c(i-2, j) \pmod p$, starting with $i = 2$, and similarly we can chase the lights to the left. 

Our elliptic curve $\tilde E$ identifies a specific subset of these infinite quiet patterns that can be used to represent and visualize points on $\tilde E$:

\begin{prop}
Let $p$ be prime. Then any point $(x, y) \pmod p$ on the curve given by $x + y + 1 + x^{-1} + y^{-1} = 0$ induces the infinite quiet pattern given by $c(i, j) = x^iy^j$ in $G$:
%I don't think the following is true:  Then there is a one-to-one correspondence between the points $\pmod p$ on the curve given by $x + y + 1 + x^{-1} + y^{-1} = 0$ and quiet patterns $c$ in $G$ such that $c(i, j)c(-i, -j) = 1$ for all $i, j \in \Z$. 
\[\begin{array}{ccccccc}
\ddots & \vdots & \vdots & \vdots & \vdots & \vdots & \iddots\\
\cdots & x^{-2}y^2 & x^{-1}y^2 & y^2 & xy^2 & x^2y^2 & \cdots \\
\cdots & x^{-2}y & x^{-1}y & y & xy & x^2y & \cdots \\
\cdots & x^{-2} & x^{-1} & 1 & x & x^2 & \cdots \\
\cdots & x^{-2}y^{-1} & x^{-1}y^{-1} & y^{-1} & xy^{-1} & x^2y^{-1} & \cdots \\
\cdots & x^{-2}y^{-2} & x^{-1}y^{-2} & y^{-2} & xy^{-2} & x^2y^{-2} & \cdots \\
\iddots & \vdots & \vdots & \vdots & \vdots & \vdots & \ddots
\end{array}
\]

\end{prop}
\begin{proof}
The equation for the elliptic curve tells us that if we start with all squares at state 0, after clicking the quiet pattern the state of each square, say at coordinates $(i, j)$, will be $x^{i+1}y^j + x^iy^{j+1} + x^iy^j + x^{i-1}y^j + x^iy^{j-1} = x^iy^j(x+y+1+x^{-1}+y^{-1}) = 0$.
\end{proof}
%Think about then plugging in roots of unity for x and y!

\eg For $p = 2$, the equation \eqref{eqn:g_proj} does not give an elliptic curve and its only solutions are listed in \cref{tab:triv_gp_a=1}; none of these actually give any solutions to $g(x, y) = 0$.

\eg For $p = 3$, \eqref{eqn:g_proj} does not give an elliptic curve and has the one nontrivial solution in the projective plane $\mathbb{P}^2$: $(x, y, z) = (2, 2, 1)$ (we always scale the $z$ coordinate to 1, so in this case this gives the point $(2, 2)$ on $g(x, y) = 0$). Indeed, if we put $c(i, j) = 2^i2^j = 2^{i+j} \pmod 3$, this gives us the following infinite checkerboard quiet pattern $\pmod 3$.
\[\begin{array}{ccccccccc}
\ddots & \vdots & \vdots & \vdots & \vdots & \vdots & \vdots & \vdots & \iddots\\ 
\cdots & 1 & 2 & 1 & 2 & 1 & 2 & 1 & \cdots \\
\cdots & 2 & 1 & 2 & 1 & 2 & 1 & 2 & \cdots \\
\cdots & 1 & 2 & 1 & 2 & 1 & 2 & 1 & \cdots \\
\cdots & 2 & 1 & 2 & 1 & 2 & 1 & 2 & \cdots \\
\cdots & 1 & 2 & 1 & 2 & 1 & 2 & 1 & \cdots \\
\cdots & 2 & 1 & 2 & 1 & 2 & 1 & 2 & \cdots \\
\cdots & 1 & 2 & 1 & 2 & 1 & 2 & 1 & \cdots \\
\iddots & \vdots & \vdots & \vdots & \vdots & \vdots & \vdots & \vdots & \ddots
\end{array}\]

\eg For $p = 5$, $\tilde E$ is still not an elliptic curve over $\F_p$ and again has one nontrivial point: $(x, y, z) = (1, 1, 1)$; if we put $c(i, j) = 1$ for all $i, j$ then we get the constant infinite quiet pattern $\pmod 5$:

\[\begin{array}{ccccccc}
\ddots & \vdots & \vdots & \vdots & \vdots & \vdots & \iddots\\ 
\cdots & 1 & 1 & 1 & 1 & 1 & \cdots \\
\cdots & 1 & 1 & 1 & 1 & 1 & \cdots \\
\cdots & 1 & 1 & 1 & 1 & 1 & \cdots \\
\cdots & 1 & 1 & 1 & 1 & 1 & \cdots \\
\cdots & 1 & 1 & 1 & 1 & 1 & \cdots \\
\iddots & \vdots & \vdots & \vdots & \vdots & \vdots & \ddots
\end{array}\]

\eg \label{eg:inf-qp-7} For $p = 7$, \eqref{eqn:g_proj} finally gives an elliptic curve, which has group $\Z/8\Z$ over $\F_p$ and the following nontrivial points: $(x, y, z) = (3, 6, 1), (6, 3, 1), (5, 6, 1), (6, 5, 1)$. This gives us four interesting quiet patterns: $c(i, j) = 3^i6^j, 6^i3^j, 6^i5^j, 5^i6^j$, all of which are symmetric to each other under the symmetries in the dihedral group $D_8$. Displayed below is the quiet pattern corresponding to $c(i, j) = 3^i6^j$.

\[\begin{array}{ccccccccc}
\ddots & \vdots & \vdots & \vdots & \vdots & \vdots & \vdots & \vdots & \iddots\\
\cdots & 1 & 3 & 2 & 6 & 4 & 5 & 1 & \cdots \\
\cdots & 6 & 4 & 5 & 1 & 3 & 2 & 6 & \cdots \\
\cdots & 1 & 3 & 2 & 6 & 4 & 5 & 1 & \cdots \\
\cdots & 6 & 4 & 5 & 1 & 3 & 2 & 6 & \cdots \\
\cdots & 1 & 3 & 2 & 6 & 4 & 5 & 1 & \cdots \\
\cdots & 6 & 4 & 5 & 1 & 3 & 2 & 6 & \cdots \\
\cdots & 1 & 3 & 2 & 6 & 4 & 5 & 1 & \cdots \\
\iddots & \vdots & \vdots & \vdots & \vdots & \vdots & \vdots & \vdots & \ddots
\end{array}\]

\eg \label{eg:inf-qp-11} For $p = 11$, $\tilde E(\F_p) \cong \Z/16\Z$ and there are only two orbits of quiet patterns under the action of $D_8$. One orbit, of size four, has representative $(1, 2, 1)$ giving the quiet pattern $c(i, j) = 2^j$ (on the left below), and the other orbit has size 8 and representative $(4, 5, 1)$ giving the quiet pattern $c(i, j) = 4^i5^j$ (on the right).

\[\begin{array}{ccccccccc}
\ddots & \vdots & \vdots & \vdots & \vdots & \vdots & \vdots & \vdots & \iddots\\
\cdots & 8 & 8 & 8 & 8 & 8 & 8 & 8 & \cdots \\
\cdots & 4 & 4 & 4 & 4 & 4 & 4 & 4 & \cdots \\
\cdots & 2 & 2 & 2 & 2 & 2 & 2 & 2 & \cdots \\
\cdots & 1 & 1 & 1 & 1 & 1 & 1 & 1 & \cdots \\
\cdots & 6 & 6 & 6 & 6 & 6 & 6 & 6 & \cdots \\
\cdots & 3 & 3 & 3 & 3 & 3 & 3 & 3 & \cdots \\
\cdots & 7 & 7 & 7 & 7 & 7 & 7 & 7 & \cdots \\
\iddots & \vdots & \vdots & \vdots & \vdots & \vdots & \vdots & \vdots & \ddots
\end{array} \qquad
\begin{array}{ccccccccc}
\ddots & \vdots & \vdots & \vdots & \vdots & \vdots & \vdots & \vdots & \iddots\\
\cdots & 9 & 3 & 1 & 4 & 5 & 9 & 3 & \cdots \\
\cdots & 4 & 5 & 9 & 3 & 1 & 4 & 5 & \cdots \\
\cdots & 3 & 1 & 4 & 5 & 9 & 3 & 1 & \cdots \\
\cdots & 5 & 9 & 3 & 1 & 4 & 5 & 9 & \cdots \\
\cdots & 1 & 4 & 5 & 9 & 3 & 1 & 4 & \cdots \\
\cdots & 9 & 3 & 1 & 4 & 5 & 9 & 3 & \cdots \\
\cdots & 4 & 5 & 9 & 3 & 1 & 4 & 5 & \cdots \\
\iddots & \vdots & \vdots & \vdots & \vdots & \vdots & \vdots & \vdots & \ddots
\end{array}
\]

We will now show how one can apply each symmetry in $D_8$ using the group law on the elliptic curve and the subgroup of trivial points on it. The point on $E_4(1)$ corresponding to $(x, y)$ is $\left(\frac{-1}{1+x+y}, \frac{-x}{1+x+y}\right)$. Using the standard formulas for the group law on an elliptic curve (see, e.g. \cite[p. 53-54]{silverman2009arithmetic}) and then using $\psi$ to come back to $E$, it is not hard to show the following.

\begin{prop}
\label{prop:sums_lead_to_triv_gp}
Let $g(x, y) = 0$. Then in the elliptic curve $E$ we have:
\begin{enumerate}[(a)]
    \item $(x, y, 1) + (y, x, 1) = (-1, 1, 0)$, i.e. $-(x, y, 1) = (y, x, 1)$
    \item $(x, y, 1) + (1/x, y, 1) = (0, 1, 0)$
    \item $(x, y, 1) + (x, 1/y, 1) = (1, 0, 0)$
    \item $(x, y, 1) + (1/y, 1/x, 1) = (0, 0, 1)$
\end{enumerate}
\end{prop}
\begin{proof}
    We prove only (d) and leave the rest to the reader as they are similar. We need to find the sum of  
    \[(x_1, y_1) = \left(\frac{-1}{1+x+y}, \frac{-x}{1+x+y}\right) \text{ and } (x_2, y_2) = \left(\frac{-1}{1+1/x+1/y}, \frac{-1/y}{1+1/x+1/y}\right) = \left(\frac{-xy}{x+y+xy}, \frac{-x}{x+y+xy}\right)\]
    on $E_4(1)$. First assume 
    \begin{equation}
    \label{eqn:assume_so_x1!=x2}
        x^2y + xy^2 - x - y = (xy - 1)(x + y) \ne 0
    \end{equation}
    so $x_1 \ne x_2$. We compute
    \[\lambda = \frac{y_2 - y_1}{x_2 - x_1} = \frac{-x}{x+y}\]
    so the $x$-coordinate of the resulting sum (using the fact that $g(x, y) = 0$) is
    \begin{align*}
        x_3 &= \lambda^2 + \lambda - 1 - x_1 - x_2\\
        &= -\frac{x^3y^2 + x^2y^3 + x^4 + 4x^3y + 7x^2y^2 + 4xy^3 + y^4 + x^2y + xy^2}{(x+y)^2(x+y+1)(xy + x + y)}\\
        &= -\frac{xy(x^2y + xy^2 + xy + x + y) + (x+y)^4}{(x+y)^2(x^2y + xy^2 + x^2 + 3xy + y^2 + x + y)}\\
        &= -\frac{(x + y)^2}{(x^2y + xy^2 + xy + x + y) + (x + y)^2}\\
        &= -\frac{(x+y)^2}{(x+y)^2} = -1.
    \end{align*}
    We also have 
    \[\nu = \frac{y_2x_1 - x_2y_1}{x_2 - x_1} = \frac{-x}{x+y} = \lambda\]
    so the $y$-coordinate of the result is
    \[-(\lambda + 1)x_3 - \nu - 1 = 0.\]
    \medskip
    If the assumption in \eqref{eqn:assume_so_x1!=x2} does not hold, then we can eliminate the possibility $x + y = 0$ since then $g(x, y) = 0$ implies $x - x + 1 + 1/x + 1/(-x) = 1 = 0$, so we must have $xy = 1$ and $g(x, y) = 0$ implies $2x^2 + x + 2 = 0$. This turns the problem into the doubling of 
    \[(x_1, y_1) = \left(\frac{-x}{x^2 + x + 1}, \frac{-x^2}{x^2 + x + 1}\right).\]
    In this case, 
    \[\lambda = \frac{3x_1^2 + 2x_1 - y_1}{2y_1 + x_1 + 1} = \frac{-x(x^2 + 2)}{(x+1)(x^2+x+1)}\]
    and the $x$-coordinate of the result is
    \[x_3 = \frac{(-1)(x^6 + 4x^5 + 2x^4 + 7x^3 + 2x^2 + 4x + 1)}{(x + 1)^2(x^2 + x + 1)^2} = \frac{(-1)[(x + 1)^2(x^2 + x + 1)^2 - (6x^4 + 3x^3 + 6x^2)]}{(x + 1)^2(x^2 + x + 1)^2} = -1.\]
    Now, some computation shows that
    \[\nu = \frac{-x_1^3 - y_1}{2y_1 + x_1 + 1} = \frac{-x^2(x^3 + x^2 + 2x + 1)}{(x-1)(x^2 + x + 1)^2}\]
    and factoring $\nu - \lambda$ gives a factor of $2x^2 + x + 2$, hence $\nu = \lambda$ and $y_3 = 0$ as before.\medskip 
    
    In either case, we get the point $(-1, 0)$ on $E_4(1)$. The corresponding point on the projectivization is $(-1, 0, 1)$, which $\psi$ takes to $(0, 0, 1)$.
\end{proof}

From (d) we have, for example, $(x, y, 1) + (0, 0, 1) = -(1/y, 1/x, 1) = (1/x, 1/y, 1)$. Similar calculations using \cref{prop:sums_lead_to_triv_gp} lead to the following corollary.

\begin{cor}
\label{cor:adding_triv_gp}
Let $g(x, y) = 0$. Then, the applications of the nontrivial symmetries in $D_8$ to the quiet pattern corresponding to $(x, y, 1)$ are given as follows: 
\begin{enumerate}[(a)]
    \item Rotation by $90^\circ$ ccw: $(1/y, x, 1) = (x, y, 1) + (0, 1, 0)$
    \item Rotation by $180^\circ$: $(1/x, 1/y, 1) = (x, y, 1) + (0, 0, 1)$
    \item Rotation by $90^\circ$ cw: $(y, 1/x, 1) = (x, y, 1) + (1, 0, 0)$
    \item Horizontal flip: $(1/x, y, 1) = (1, 0, 0) - (x, y, 1)$
    \item Vertical flip: $(x, 1/y, 1) = (0, 1, 0) - (x, y, 1)$
    \item Flip about wrong diagonal (/): $(y, x, 1) = -(x, y, 1)$
    \item Flip about main diagonal ($\backslash$): $(1/y, 1/x, 1) = (0, 0, 1) - (x, y, 1)$.
\end{enumerate}
\end{cor}

For any elliptic curve with a subgroup $H$ of $n$ points, we can consider the orbit of a point $P$ under $D_{2n}$ given by $(H + P) \cup (H - P)$; in our case, when $H$ is the subgroup $\Z/4\Z$ shown in \cref{tab:triv_gp_a=1}, this orbit can actually be visualized as the literal action of $D_8$ by its natural symmetries on infinite quiet patterns.

\section{Further Properties}
\label{sect:further}
\begin{prop}
\label{prop:perfect}
Let $p$ be prime. The evolution $\Evo_p(Q)$ of a pattern that has no $\vec 0$ column or row produces a pattern that also has no $\vec 0$ column or row, unless $p = 2$ and two adjacent columns or rows of $Q$ are the same. Conversely, if $Q$ has a $\vec 0$ column or row, so does $\Evo(Q)$.
\label{prop:PerfectionOfEvo}
\end{prop}
\begin{proof}
We will prove that $\Evo_p(Q)$ has no $\vec 0$ column given that the same is true for $Q$ and either $p > 2$ or there is an odd number of columns; of course, the same argument applies equally well for rows. Label the columns of $Q$ from 1 to $N$ and those of $\Evo(Q)$ from $1$ to $pN+p-1$. Any column with index $pi$ in $Q$ contains the centers of evolutions of single squares, so we can just take the top nonzero click in the corresponding column in $Q$ (namely the one with index $i$)s and note that the top half of the middle column of the evolution of that square cannot be cancelled out by anything in $\Evo(Q)$; in particular, the top square of that middle column has the same nonzero number of clicks as the corresponding square in $Q$. 

If $p$ is odd, then exactly half of the $p-1$ columns on either side of that middle column cannot be cancelled out by anything in $\Evo(Q)$, so the same top nonzero entry in Column $i$ of $Q$ also guarantees that Columns $pi \pm j$ in $\Evo(Q)$ are not $\vec 0$, for $j = 1, 2, \dots, (p-1)/2$. This argument shows that none of the columns can be $\vec 0$, except possibly the first or last $(p-1)/2$, and those are easily taken care of by considering the top nonzero entries in the first and last column of $Q$ (for example, the one in the column 1 of $Q$ has an evolution such that the top half of the first $p$ columns cannot be cancelled by anything in $\Evo(Q)$).

If $p = 2$ then similarly, the first 2 and last 2 columns of $\Evo(Q)$ cannot be $\vec 0$, and (as already proven above) neither can any column with even index. A column of index $2i + 1$ would be $\vec 0$ in $\Evo(Q)$ if and only if all the middle 1's from the ``plusses'' centered at columns $2i$ and $2i+2$ all cancel, which happens precisely when columns $i$ and $i+1$ are the same in $Q$.

Conversely, it is clear that if the $i$-th column of $Q$ is $\vec 0$, so is the $pi$-th column of $\Evo(Q)$.
\end{proof}
\rmk If $p=2$ and $Q$ is a quiet pattern, then $Q$ cannot have two identical columns unless the number of columns in $Q$ is even. If there are two identical columns in $Q$, we can starting from the two equal columns and proceed to ``chase the lights'' to the left and right from those; the same columns follow on both sides, and if $Q$ is a quiet pattern, this chasing stops when we reach a column where don't have to click at all, and such a column must be reached at the same time on the left and right sides, making the number of columns even.

A simple corollary of \cref{prop:GenMikado} is the following. Let $\bar G$ be an infinite grid of squares that is bounded only at the top. Suppose we wanted to make a quiet pattern in $\bar G$ where exactly one square is clicked in the top row, say $k_0$ times. Then there is a unique way to ``chase the lights'' down from the top row, so such a quiet pattern would be unique up to horizontal translation; we will call the pattern $\bar Q(p, k_0)$ when the number of possible states is $p$.
\begin{prop}
Let $p$ be prime. The top $p^n$ rows of $\bar Q(p, k_0)$ are those in $\Evo_p^{(n)}(k_0)$.
\label{prop:infProcDown}
\end{prop}

Christopher Ratigan had made the following conjecture as an undegraduate. Consider the ``connected components'' of squares that are clicked in a quiet pattern with 2 colors, i.e. connected components of the subgraph of the grid graph (where two vertices are adjacent if they share an edge) consisting of those squares that are actually clicked. In \cref{fig:qp2x3} we see that the left two quiet patterns are connected, and we will call their shape a ``Tetris T''; the right quiet pattern in the same figure shows two vertical ``dominoes'' of 1's side-by-side, separated by an empty column.

\begin{prop}[Ratigan]
\label{prop:Ratigan}
The only connected components of $\bar Q(2, 1)$ are Tetris ``T''s and dominoes.
\end{prop}
\begin{proof}
First, note that the evolution of a pattern with two separate connected components always has the parts corresponding to the two components separated (if two components are diagonally adjacent, there will be a diagonal of 0's between their evolutions, and if they are separated by a line then their evolutions will be too). The evolution of either a domino or a ``T'' is composed of connected components that are again just dominoes and T's, as can be seen in \cref{fig:qp5x7}. The evolution of a single square with a 1 in it, i.e. $\Evo_2(1)$, is a `+' shape (see \cref{eg:Evo_2(1)}), and the evolution of a `+' shape is the following:
\[\Evo_2^{(2)}(1) = \begin{array}{ccccccc}
0 & 0 & 0 & 1 & 0 & 0 & 0 \\
0 & 0 & 1 & 1 & 1 & 0 & 0 \\
0 & 1 & 0 & 0 & 0 & 1 & 0 \\
1 & 1 & 0 & 1 & 0 & 1 & 1 \\
0 & 1 & 0 & 0 & 0 & 1 & 0 \\
0 & 0 & 1 & 1 & 1 & 0 & 0 \\
0 & 0 & 0 & 1 & 0 & 0 & 0 \\
\end{array}\]
so the only connected components of the patterns $\Evo_2^{(n)}(1)$ are dominoes and T's, with the exception of the center, which alternates between a single square and a `+' shape, and in particular the top $2^n - 2$ rows consist of nothing but dominoes and T's.
\end{proof}

We now move on to understanding how many squares can have a nonzero state if we start clicking from the all-zero state, and only click a finite number of times, in the grid of squares $G$ from \cref{sect:EllCurves} that stretches infinitely far in all directions. If we click in the alternating pattern $1, -1, 1, \dots$ consecutively $r$ times along a diagonal, we make exactly $r + 4$ squares have a nonzero state, so it is possible to have exactly $n$ squares with a nonzero state for any $n \ge 5$.

Assuming all squares start in a state of 0 as always, we will say the light in a square is ``on'' or ``lit'' if the square has a nonzero state as a result of some clicks.

\begin{thm}
    \label{thm:5Lights}
        Let $p$ be prime. Suppose we have a grid $G$ stretching infinitely far in all four directions. After a finite number of clicks anywhere in $G$, the minimum number of lit squares will always be at least five, and the only way to have exactly five squares lit is to click in an $\Evo_p^{(n)}(k_0)$ pattern.
        %\item For a composite number of colors $k$, we can pick any prime $p \mid k$ and click in a $(kk_0/p)\Evo_p^{(n)}(1)$ pattern for any $n$ and any $k_0 = 1, 2, \dots, p-1$ to leave only five lights on, and this gives the only ways to leave only five squares lit.
        %\item When working over $\Z$, we can leave five lights on only if we click one square.
\end{thm}

The proof starts the same as in \citep{Mikado}. Note that if our clicks are restricted to a line (horizontal, vertical, or diagonal) and there is more than one click, then (considering the endpoints of the line) the number of lit squares will be at least 6, so we may assume the clicks are not restricted to a line. 

If we consider the smallest rectangle that contains all the clicks, after the clicks there must be a nonzero state on each side of the rectangle, just outside of the rectangle, so there cannot be any less than 4 squares with nonzero states. 
A pattern that leaves only 4 squares with nonzero states must have a nonzero number of clicks in only one square in the top and bottom rows and the leftmost and rightmost columns of this bounding rectangle. The pattern must then proceed at a $45^\circ$ angle from each of these squares, alternating signs as it goes. %($c$ clicks, then $-c$, then $c$, then $-c$, and so on). Let $k = p$ be prime first. 
We assume without loss of generality that the leftmost number of clicks is $c = 1$ (else divide the whole pattern by $c$). There is only one way to turn off the leftmost square clicked, then the square above and right of that, and then the square above-right of that and so on, so we find that there is a unique way that the diagonal just below must be clicked (the diagonal consisting of $-1, 2, -3, 4, \dots$ in  \cref{fig:minimal_pattern}). Then, we can continue chasing lights in this way and we find there is a unique way to click the diagonal below, each number of clicks being forced by the fact that the square to the left must be turned off, so we get the diagonal $-1, 2, -4, 7, -11, \dots$ then $2, -5, 10, -18, 30, \dots$, next $1, -4, 11, -24, 46, \dots$, and so on. These form a triangle which is the OEIS sequence A364366 with alternating signs.

\begin{figure}[hbtp]
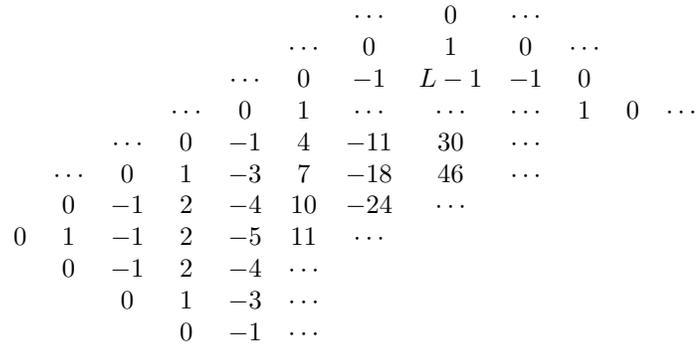

    \centering
    $\begin{array}{ccccccccccccc}
&&&&&&\cdots&0&\cdots\\
&&&&&\cdots&0&1&0&\cdots\\
&&&&\cdots&0& -1 & L-1 & -1&0\\
&&&\cdots&0&1&\cdots & \cdots & \cdots & 1 & 0 & \cdots\\
&&\cdots&0&-1&4&-11&30 &\cdots\\
&\cdots&0&1&-3&7&-18&46&\cdots\\
&0&-1&2&-4&10&-24&\cdots\\
0&1&-1&2&-5&11&\cdots\\
&0&-1&2&-4&\cdots\\
&&0&1&-3&\cdots\\
&&&0&-1&\cdots
\end{array}$
    \caption{Diagonals formed by chasing from bottom-left to the upper-right.}
    \label{fig:minimal_pattern}
\end{figure}

Suppose the length of the diagonal is $L$. Since the state of the topmost square must end up 0, we must have $(L-1) + 1 \equiv 0 \pmod p$, and similarly each of the diagonals in the rectangular array of clicks must have length divisible by $p$. Now, for primes $p$ for which $c_p \ne 0$, the proof in \citep{Mikado} generalizes in a fairly straightforward way, as we encourage the reader to verify, though there is a point not discussed in \citep{Mikado} which we will work out in the very end of the proof. And though such primes have density 0 (see \cref{thm:c_p_ell_curve}), we still want to prove this for all primes $p$.

We have shown that the length of the outermost diagonal must be divisible by $p$ or else the topmost square clicked will remain lit. Now we will generalize this statement by examining how long we can continue chasing along a length-$L$ diagonal and keep making the squares in the diagonal above turn to a state of 0. First note that we have ``longer'' diagonals of clicks, say of length $n$ as above, intermixed with ``shorter'' diagonals in between them (of length $n-1$). As we chase the lights by clicking in a ``shorter'' diagonal, if the last square clicked in that diagonal does not turn the state of the square above it to 0, then the process will fail. In other words, if we chase the lights to the right of a ``longer'' diagonal, we must get 0 clicks at the very end. The pattern $\Evo_p^{(n)}$ tells us exactly how to chase the lights along diagonals of length $p^{n}$, so let us take a look at what this pattern can show us, as illustrated in \cref{fig:evo_3-diamonds}. 

For one, the pattern always has $\pm 1$'s around the outside, since the outermost clicks are $1$ (the coefficient of $x^{p-1}$ in $g(x, y)^{p-1}$ is 1) and the resulting state after clicking the pattern has no lights on around the edge of the diamond which makes up its nonzero entries. 
% This is our favorite way of showing that
% \[{p-1 \choose i} \equiv (-1)^i \pmod p,\]
% since it is the coefficient  of $x^iy^{p-1-i}$ in $(x + y + 1 + x^{-1} + y^{-1})^{p-1}$ and these are precisely the numbers along the edge of the pattern.
%MAKE FIGURE HERE!!
As the pattern evolves from $\Evo_p(1)$ to $\Evo_p^{(n+1)}(1)$, these $\pm 1's$ are replaced by $\Evo_p^{(n)}(1)$ patterns themselves, which are positioned as adjacent diamonds. Consider the shorter diagonal between any two such adjacent diamonds. Because there are only 0's outside the diamond pattern of nonzero clicks that is $\Evo_p(1)$, counting from the outside inward, there are $(p^n-1)/2$ zeroes along these shorter diagonals ($(3^1-1)/2 = 1$ zero in \cref{fig:evo_3-diamonds}). These zeroes are followed by $\pm k_0$ coming from the next diamond, which is a pattern $\Evo_p^{(n)}(k_0)$ for some number $k_0$ in the second (shorter) diagonal of the primitive pattern $\Evo_p(1)$. Now any number $k_0$ in the second diagonal of the $\Evo_p(1)$ pattern must be nonzero $\pmod p$: it is the coefficient of $x^{p-2-i}y^i$ in $g(x, y)^{p-1} = (x + y + 1 + x^{-1} + y^{-1})^{p-1}$ for some $i = 0, 1, \dots, p-2$, and in the expansion of $g(x, y)^{p-1}$, such a term can only be found by multiplying $p-2-i$ factors equal to $x$, $i$ factors equal to $y$, and one factor equal to 1, making the coefficient $(p-1)!/[(p-2-i)!i!1!]$, which cannot be divisible by $p$.

With this understanding, if the length $L$ of the diagonal we are chasing along is evenly divisible by $p^n$ (i.e. but not by $p^{n+1}$), then we know exactly how chasing the lights along that diagonal works by looking at a pattern $\Evo_p^{(N)}(1)$ for some $N$ such that $p^N > L$. The fractal nature of this pattern tells us that the edge of our rectangle falls along the edge of an $\Evo_p^{(n)}$ diamond of clicks, and there are $(p^n-1)/2$ zeroes between this diamond and the adjacent diamond, meaning we can successfully chase along $(p^n-1)/2$ shorter diagonals and along $(p^n-1)/2 + (p^n-1)/2 + 1$ totals diagonals while leaving the states of all squares in the previous diagonals 0, but the next shorter diagonal will not work because the number of clicks after its end is $\pm k_0$ for some $k_0 \ne 0$ in the second diagonal of $\Evo_p(1)$, which means the square above the end of this shorter diagonal will not be turned to a state of 0 via this chasing. Hence, we have

\begin{lem}
    \label{lem:chasing_limit}
    Let $p$ be prime. If a side of a diagonal rectangle of clicks is divisible by $p^n$ but not $p^{n+1}$, then chasing lights parallel to this diagonal will fail at the $(p^n+1)$st diagonal (we include the shorter diagonals in the count).
\end{lem}

As a corollary, if the side of a diagonal rectangle has length $L$ and an adjacent side is longer than $(L+1)/2$, then we must have a square lit after the clicks in this rectangle, and the only way we could have no squares lit after chasing parallel to $L$ is if $L$ is a power of $p$ and both sides adjacent to $L$ are at most $(L+1)/2$ in length. Of course, it doesn't matter along which side we chase the lights, so we immediately see that if we chase along the shorter side in this case, we would have to leave a light on again.

\begin{proof}[Proof of Theorem \ref{thm:5Lights}]
    First we prove we cannot leave only 4 lights on. Suppose two adjacent sides of the rectangle of clicks have lengths $a$ and $b$. Then, chasing along the diagonal of length $b$, we see that a light will be left inside the rectangle or on its side (not outside) unless $a \le (b+1)/2$. Similarly, chasing along the diagonal of length $a$, we must have $b \le (a+1)/2$, so $a \le (a+3)/4$, which implies $a \le 1$, and similarly $b \le 1$, so there is just the one square clicked (which leaves 5 lights on) - a contradiction.

    Now, suppose we have only 5 lights left on, and only one square is clicked on each edge of the bounding rectangle (so only one highest square, one lowest, one leftmost, and one rightmost are clicked). Then the 5th light that remains lit must be in the very middle of the rectangle and each diagonal must be length $p^n$ for some $n$ because no matter what side we start chasing along, we won't be able to leave all lights off after passing the $p^n$-th (or middle) diagonal. Since there is a unique way to chase the lights, and the $\Evo_p^{(n)}(1)$ pattern accomplishes this, the clicks must be this pattern.

    One more possibility remains: we could have the 5 lights be on outside of the rectangle containing the clicks, namely we could have 2 nonzero clicks on one edge, a point not made in \citep{Mikado}. This can't happen for 2-color puzzles because you could perform the erasing procedure (see \cref{thm:Era_p}) until the distance between the 2 squares clicked on one edge is odd, so the diagonals on the opposite side don't line up either and must end up coming together an odd distance away on the opposite side, which means they don't come to one place on the opposite side and we have at least 6 lights on.

    For $p > 2$ colors, suppose we have two nonzero clicks in two different squares in the top edge and only one nonzero click in the rest of the edges, so all the lit squares as a result of the clicks are outside of the rectangle. Note that we must proceed diagonally down-right from the upper-left nontrivial click for, say, $b$ clicks and down-left for, say, $c$ clicks, and say from the left edge we proceed down-right for $d$ clicks. Then, by \cref{lem:chasing_limit}, if we try and chase the lights along the side with length $d$, we see that the adjacent side must satisfy $c \le (d+1)/2$, lest a square inside or along the edge of the rectangle of clicks remain lit. But we could also chase parallel to the side of length $c$ to find that the adjacent side must satisfy $b \le (c+1)/2$, so we have $b \le (d+3)/4$ and $b + c \le (3d+5)/4$. The same is true on the right side; if we label the corresponding side lengths on the right as $b', c', d'$ then $b' + c' \le (3d'+ 5)/4$. But, counting the total number of columns in the pattern, we have
    \[d + d' - 1 = b + c + b' + c' - 3,\]
    so
    \[d + d' = b + c + b' + c' - 2 \le \frac34(d + d') + 1/2,\]
    which says $d + d' \le 2$, meaning all the squares are in one column: a contradiction. 

    % Now (d) is clear because any pattern that leaves exactly 5 lights on must be an evolution pattern $\pmod p$ for \textit{every} prime $p$, and the only pattern all these evolution patterns have in common is the single click (every pattern has a diagonal side length that is a power of $p$).

    % Now for (c) and to finish the proof of (a), suppose $k$ is not prime. If every number of clicks in the pattern is divisible by a prime $p$, we can divide the pattern and $k$ by $p$, and thus assume that the gcd of all the numbers in the pattern is 1. Chasing lights works the same way $\pmod p$ for every prime $p \mid k$. By the proof for primes, for a given prime $p$, this means the pattern is an evolution pattern $\pmod p$, and that the pattern leaves exactly 5 lights on. Obviously the pattern cannot be an evolution pattern for two primes simultaneously %(we can't have p^n = q^m for two different primes p, q)
    % so $k$ must be a prime power at this point. 

\end{proof}

\section{Acknowledgements}
We are indebted to Patrick Morton for numerous insightful conversations and generous sharing of time and expertise. We are also grateful for the intellectual support and encouragement of Fariba Khoshnasib-Zeinabad. This paper was started as undergraduate research, funded by Earlham College.

\bibliographystyle{plainnat}
\bibliography{refs}

\end{document}